\documentclass[reqno, 10pt]{amsart}

\usepackage{amsmath, amsfonts, a4wide}
\usepackage[colorlinks=true]{hyperref}

\title[Reducible connections and non-local symmetries]
{Reducible connections and non-local symmetries of the self-dual Yang--Mills equations}
\author{James~D.E.\ Grant}
\date{22 December, 2009}

\numberwithin{equation}{section}

\theoremstyle{plain}
\newtheorem{theorem}{Theorem}[section]

\newtheorem{proposition}[theorem]{Proposition}

\theoremstyle{definition}
\newtheorem{definition}[theorem]{Definition}
\newtheorem{example}[theorem]{Example}
\newtheorem*{notation}{Notation}

\theoremstyle{remark}
\newtheorem{remark}[theorem]{Remark}

\newcommand{\sdyme}{self-dual Yang--Mills equations}
\newcommand{\cpone}{\ensuremath{\mathbb{C}P^1}}
\newcommand{\cptwo}{\ensuremath{\mathbb{C}P^2}}
\newcommand{\cpthree}{\ensuremath{\mathbb{C}P^3}}

\newcommand{\Id}{\ensuremath{\mathrm{Id}}}

\newcommand{\Vzeroeps}{\ensuremath{\mathcal{V}^0_{\epsilon}}}
\newcommand{\Vinfeps}{\ensuremath{\mathcal{V}^{\infty}_{\epsilon}}}
\newcommand{\Veps}{\ensuremath{\mathcal{V}_{\epsilon}}}

\newcommand{\SLtwoC}{\ensuremath{\mathrm{SL}_2(\mathbb{C})}}
\newcommand{\sltwoC}{\ensuremath{\mathfrak{sl}_2(\mathbb{C})}}
\newcommand{\PSLtwoC}{\ensuremath{\mathrm{PSL}_2(\mathbb{C})}}

\newcommand{\SLtwoR}{\ensuremath{\mathrm{SL}_2(\mathbb{R})}}
\newcommand{\PSLtwoR}{\ensuremath{\mathrm{PSL}_2(\mathbb{R})}}

\newcommand{\SUtwo}{\ensuremath{\mathrm{SU}_2}}
\newcommand{\sutwo}{\ensuremath{\mathfrak{su}_2}}

\newcommand\ubar{\ensuremath{\overline{u}}}
\newcommand\vbar{\ensuremath{\overline{v}}}

\newcommand\bracket[2]{\ensuremath{\left\{ #1 \, \left| \vphantom{|^|} \, #2 \right.\right\}}}

\newcommand{\vs}{\vskip .2cm}

\begin{document}
\dedicatory{To David~E.\ Williams}
\thanks{This work was supported by START-project Y237--N13 of the Austrian Science Fund and a Visiting Professorship at the University of Vienna. The author is grateful to Prof.~M.A.\ Guest for discussions and, in particular, for drawing his attention to harmonic maps of finite type. He would also like to thank the anonymous referee for a detailed critique of the original version of this paper.}
\address{\href{http://www.mat.univie.ac.at/home.php}{Fakult{\"a}t f{\"u}r Mathematik} \\
\href{http://www.univie.ac.at/de/}{Universit{\"a}t Wien} \\ Nordbergstrasse 15 \\ 1090 Wien \\ Austria}
\email{\href{mailto:james.grant@univie.ac.at}{james.grant@univie.ac.at}}

\begin{abstract}
We construct the most general reducible connection that satisfies the self-dual Yang--Mills equations on a simply-connected, open subset of flat $\mathbb{R}^4$. We show how all such connections lie in the orbit of the flat connection on $\mathbb{R}^4$ under the action of non-local symmetries of the self-dual Yang--Mills equations. Such connections fit naturally inside a larger class of solutions to the self-dual Yang--Mills equations that are analogous to harmonic maps of finite type.
\end{abstract}
\maketitle
\thispagestyle{empty}

\section{Introduction}

Reducible connections play an important r\^{o}le in Donaldson's study of four-manifold topology~\cite{Don1}. In particular, the singular ends of the moduli space of global solutions of the {\sdyme} on a four-manifold are due to the existence of connections on which the group of gauge transformations (modulo its centre) do not act freely. In the current paper, we study reducible connections from a different point of view, that of integrable systems theory.

In this paper and its companion~\cite{sdym} we investigate the non-local symmetry algebra of the {\sdyme} on $\mathbb{R}^4$ discussed in~\cite{CGSW,CGW,CW} and corresponding group actions on spaces of solutions of the {\sdyme} on open subsets of $\mathbb{R}^4$. In~\cite{sdym} we studied instanton moduli spaces, as explicitly described by the ADHM construction~\cite{At1,ADHM}, and group actions that preserved the $L^2$ nature of the curvature of the connection. In the current work, we investigate reducible connections defined on a simply-connected, open subset of $\mathbb{R}^4$. Given that reducible and irreducible connections play a different r\^{o}le in Donaldson's work, our main motivation was to study whether such connections have different properties from an integrable systems point of view. This does, indeed, seem to be the case.

In the case of instanton solutions on $\mathbb{R}^4$, it was argued in~\cite{sdym} that the symmetry group that acts on the moduli space has orbits of high codimension in the moduli space. (In other words, the orbits are quite small.) Our conclusion for reducible connections, however, are quite different. After explicitly constructing the most general reducible self-dual connection on a simply-connected, open subset of $\mathbb{R}^4$ (there are no reducible self-dual connections on $\mathbb{R}^4$ with $L^2$ curvature), we deduce that all reducible connections lie in the orbit of the flat connection on $\mathbb{R}^4$ under the action of the non-local symmetry group of the {\sdyme}. Also, the reducible connections lie within a larger class of solutions that arise quite naturally from the symmetries of the {\sdyme}. Solutions in this larger family are determined by a holomorphic function, $T$, defined on an open subset of {\cpthree} that obeys the condition $\left[ T, T^* \right] = 0$. Such formulae bear a strong resemblance to those arising in the theory of harmonic maps of finite type (see, e.g., \cite[Chapter~24]{Gu}). This result is distinct from the instanton case discussed in~\cite{sdym}, which bears more of a resemblance to the theory of harmonic maps of finite uniton number~\cite{UhlenbeckJDG}.

\vs
The organisation of this paper is as follows. In the following Section, we set up notation and recall the non-local symmetries of the {\sdyme} on $\mathbb{R}^4$ constructed in~\cite{CGSW,CGW,CW}. We also recall the main results of~\cite{Cr} concerning the twistorial interpretation of these symmetries in terms of their action on the patching matrix of holomorphic bundles over open subsets of {\cpthree}. In Section~\ref{sec:reducible}, we determine the most general reducible self-dual connection on a simply-connected, open subset of $\mathbb{R}^4$. We show that these may be constructed directly from harmonic functions. We then show the patching matrix of such connections may be constructed directly. In Section~\ref{flatorbit}, we deduce that all such patching matrices, and therefore all reducible self-dual connections, lie in the orbit of the flat connection on $\mathbb{R}^4$. In particular, we see there is a larger class of patching matrices that appear quite naturally from the group action of~\cite{Cr} that contains all reducible connections. Analogies between this larger class of solutions and harmonic maps of finite type are briefly investigated in Section~\ref{sec:HMFT}. After some final remarks, in an Appendix we study some properties of a simplified version of the group action of~\cite{Cr}.

\vs
As in the companion paper~\cite{sdym}, we study only the non-local symmetries of the {\sdyme} constructed in~\cite{CGSW,CGW,CW}, and not symmetries that require the existence of a non-trivial conformal group on our manifold. We also specialise to the case of {\SUtwo} structure group, although the generalisation to any classical Lie group is straightforward.

\section{Preliminaries}

\subsection{The self-dual Yang--Mills equations on $\mathbb{R}^4$}
Let $U$ be a connected, simply connected open subset of $\mathbb{R}^4$ with its flat metric. From the Cartesian coordinates $(t, x, y, z)$ on $\mathbb{R}^4$, we define complex coordinates
\[
u := t + i x, \qquad
v := y - i z
\]
on $\mathbb{R}^4 \cong \mathbb{C}^2$. In terms of these coordinates, the metric on $\mathbb{R}^4$ is
\[
\mathbf{g} = \frac{1}{2} \left( du \otimes d{\ubar} + d{\ubar} \otimes du +
dv \otimes d{\vbar} + d{\vbar} \otimes dv \right).
\]
and the standard volume form is
\[
\epsilon = dt \wedge dx \wedge dy \wedge dz =
\frac{1}{4} du \wedge d{\ubar} \wedge dv \wedge d{\vbar}.
\]
In terms of these coordinates, a local basis for the bundle of anti-self-dual two-forms (with respect to the above metric and volume form) on $U$ is given by $\{ du \wedge dv, \frac{1}{2} \left( du \wedge d\ubar + dv \wedge d\vbar \right), d\ubar \wedge d\vbar \}$.

Let $\pi\colon P \rightarrow U$ be a principal {\SUtwo} bundle over $U$. Since we are, essentially, working locally, a connection on $P$ may be represented by an {\sutwo}-valued one-form $\mathbf{A} \in \Omega^1(U, \sutwo)$, with curvature $\mathbf{F} \in \Omega^2(U, \sutwo)$. In terms of the complex coordinates above, the connection satisfies the self-dual Yang--Mills equations on $U$ if and only if
\begin{subequations}
\begin{gather}
F_{uv} = 0,
\label{Fuv}
\\
F_{u{\ubar}} + F_{v{\vbar}} = 0,
\label{Fuubar}
\\
F_{{\ubar}{\vbar}} = 0.
\label{Fubarvbar}
\end{gather}\label{sdym}\end{subequations}

Since $U$ is simply-connected, \eqref{Fuv} and~\eqref{Fubarvbar} imply the existence of maps $\psi_0, \psi_{\infty}\colon U \rightarrow \SLtwoC$ with the property that
\begin{subequations}
\begin{gather}
A_u = - \left( \partial_u \psi_{\infty} \right) \psi_{\infty}^{-1}, \qquad
A_v = - \left( \partial_v \psi_{\infty} \right) \psi_{\infty}^{-1}, \\
A_{\ubar} = - \left( \partial_{\ubar} \psi_0 \right) \psi_0^{-1}, \qquad
A_{\vbar} = - \left( \partial_{\vbar} \psi_0 \right) \psi_0^{-1}.
\end{gather}\label{Acomponents}\end{subequations}
The fields $\psi_0, \psi_{\infty}$ are determined by equations~\eqref{Acomponents} up to transformations
\[
\psi_0(x) \mapsto {\widetilde\psi}_0(x) := \psi_0(x) R(u, v), \qquad
\psi_{\infty}(x) \mapsto {\widetilde\psi}_{\infty}(x) := \psi_{\infty}(x) S(\ubar, \vbar),
\]
where $R$, $S$ are arbitrary analytic functions of $(u, v)$, $(\ubar, \vbar)$ respectively. We may use this freedom to set, without loss of generality, $\psi_{\infty}(x) = \left( \psi_0(x)^{-1} \right)^{\dagger}, \forall x \in U$. The remaining freedom in the choice of $\psi_0, \psi_{\infty}$ is then of the form
\begin{equation}
\psi_0(x) \mapsto {\widetilde\psi}_0(x) := \psi_0(x) R(u, v), \qquad
\psi_{\infty}(x) \mapsto {\widetilde\psi}_{\infty}(x) := \psi_{\infty}(x) \left( R(u, v)^{-1} \right)^{\dagger}.
\label{psifreedom}
\end{equation}
In terms of these fields we define the \emph{Yang $J$-function}, $J\colon U \rightarrow \SLtwoC$, by
\[
J(x) := \psi_{\infty}^{-1}(x) \cdot \psi_0(x), \qquad x \in U.
\]
It follows from the reality properties of $\psi_0, \psi_{\infty}$ that $J$ is Hermitian
\[
J(x) = J(x)^{\dagger}, \qquad x \in U,
\]
and that, under the transformation~\eqref{psifreedom}, $J$ transforms according to the rule
\begin{equation}
J(x) \mapsto {\widetilde J}(x) := R(u, v)^{\dagger} J(x) R(u, v).
\label{Jtransformation}
\end{equation}

Substituting into equation~\eqref{Fuubar}, we find that the connection, $\mathbf{A}$, satisfies the {\sdyme} if and only if $J$ satisfies the two (equivalent) versions of the \emph{Yang--Pohlmeyer equation}
\begin{subequations}
\begin{gather}
\partial_u \left( J_{\ubar} J^{-1} \right) + \partial_v \left( J_{\vbar} J^{-1} \right) = 0,
\label{YP1}
\\
\partial_{\ubar} \left( J^{-1} J_u \right) + \partial_{\vbar} \left( J^{-1} J_v \right) = 0.
\label{YP2}
\end{gather}\label{Jeqn}\end{subequations}

\subsection{Associated linear problem}
Let $\Omega$ be a non-empty, open subset of $\cpone := \mathbb{C} \cup \{ \infty \}$, and consider the following overdetermined system of equations for a map $\Psi\colon U \times \Omega \rightarrow \SLtwoC$
\begin{subequations}
\begin{align}
\left( \partial_{\vbar} + z \partial_u \right) \Psi(x, z) &=
- \left( A_{\vbar} + z A_u \right) \Psi(x, z),
\\
\left( \partial_{\ubar} - z \partial_v \right) \Psi(x, z) &=
- \left( A_{\ubar} - z A_v \right) \Psi(x, z),
\\
\partial_{\overline{z}} \Psi(x, z) &= 0.
\end{align}\label{ALP}\end{subequations}
An important property of the {\sdyme} is that they are the integrability condition for this system. In particular, if the connection $\mathbf{A}$ satisfies the {\sdyme} on $U$, then there exists $\epsilon > 0$ and a solution $\Psi_0\colon U \times \Vzeroeps \rightarrow \SLtwoC$ of this system that is analytic in $z$ for $z \in \Vzeroeps$, where $\Vzeroeps := \bracket{z \in \cpone}{|z| < 1 + \epsilon}$.

\begin{notation}
We define the involution $\sigma\colon \cpone \rightarrow \cpone$ by $\sigma(z) = - \left. 1 \right/\!\overline{z}$. Given a subset $\mathcal{V} \subset \cpone$ and a map $g\colon \mathcal{V} \rightarrow \SLtwoC$, we define a corresponding map $g^*\colon \sigma(\mathcal{V}) \rightarrow \SLtwoC$ by
\[
g^*(z) := \left( g(\sigma(z)) \right)^{\dagger}.
\]
Similarly, given any map $f\colon U \times \mathcal{V} \rightarrow \SLtwoC$, we define a corresponding map $f^*\colon U \times \sigma(\mathcal{V}) \rightarrow \SLtwoC$ by
\[
f^*(x, z) := \left( f(x, \sigma(z)) \right)^{\dagger}.
\]
\end{notation}

Given the solution, $\Psi_0$, of~\eqref{ALP}, we may now construct another solution $\Psi_{\infty}\colon U \times \Vinfeps \rightarrow \SLtwoC$, where $\Vinfeps := \bracket{z \in \cpone}{|z| > \frac{1}{1+\epsilon}}$, by $\Psi_{\infty}(x, z) := \Psi_0^*(x, z)^{-1}$. The solution $\Psi_{\infty}$ is analytic in $z$ for $z \in \Vinfeps$.

\begin{remark}
Note that, for the construction of the connection in equation~\eqref{Acomponents}, we may take $\psi_0(x) := \Psi_0(x, 0)$ and $\psi_{\infty}(x) := \Psi_{\infty}(x, 0)$. We will assume, from now on, that $\psi_0$ and $\psi_{\infty}$ are defined in this way.
\end{remark}

\begin{definition}
The \emph{patching matrix} (or \emph{clutching function}, in Uhlenbeck's terminology~\cite{UhlenbeckJDG}), $G\colon U \times \Veps \rightarrow \SLtwoC$ is defined by
\begin{equation}
G(x, z) = \Psi_{\infty}(x, z)^{-1} \cdot \Psi_0(x, z),
\label{patching}
\end{equation}
where $\Veps := \Vzeroeps \cap \Vinfeps = \bracket{z \in \cpone}{\frac{1}{1+\epsilon} < |z| < 1+\epsilon}$.
\end{definition}

\begin{remark}
Viewing $U \times \Veps$ as a subset of $\pi^{-1}(U) \subseteq \cpthree$, the patching matrix is the transition function of the holomorphic vector bundle over {\cpthree} corresponding to our self-dual connection $\mathbf{A}$~\cite{AHS,Wa}. Since $U \times \Vzeroeps$ and $U \times \Vinfeps$ are open subsets of $\mathbb{C}^3$, any holomorphic bundle over them is trivial. As such, the bundle over $\pi^{-1}(U)$ is completely determined by the patching matrix $G$ (see, e.g., \cite{Cr}). The fact that the patching matrix splits as in~\eqref{patching} implies that the bundle is trivial on restriction to a line $\pi^{-1}(x)$, for each $x \in U$~\cite{Wa}. Since the patching matrix obeys the reality condition $G(t, z) = G^*(t, z)$, the bundle admits a Hermitian structure, and the corresponding self-dual connection is an {\SUtwo} connection, rather than an {\SLtwoC} connection.
\end{remark}

\subsection{Non-local symmetries}
In order to study symmetries of the {\sdyme}, we let $J(\cdot, s)\colon U_s \rightarrow \SLtwoC$ be a one-parameter family of solutions of the Yang--Pohlmeyer equations~\eqref{Jeqn}. Here, $s \in I$ with $I$ an open interval in $\mathbb{R}$ containing the origin, $J$ is assumed to depend in a $C^1$ fashion on the parameter $s$, and $U_s \subseteq \mathbb{R}^4$ is the open subset of $\mathbb{R}^4$ on which the solution is well defined (i.e. non-singular)\footnote{We will often notationally suppress the dependence of the domain $U$ on $s$.}. Taking the derivative with respect to $s$ of~\eqref{Jeqn}, we find that $J(\cdot, s)$ must satisfy the linearised equation
\begin{equation}
\partial_u \left( J \partial_{\ubar} \left( J^{-1} \dot{J} \right) J^{-1} \right) +
\partial_v \left( J \partial_{\vbar} \left( J^{-1} \dot{J} \right) J^{-1} \right) = 0.
\label{linearisation}
\end{equation}
where $\dot{J} := \left. \partial J \right/\! \partial s$. It is known that the only local symmetries of the {\sdyme} on flat $\mathbb{R}^4$ are gauge transformations and those generated by the action of the conformal group (see, e.g.,~\cite{Po}). However, there exists a non-trivial family of non-local symmetries of the {\sdyme}~\cite{CGSW,CGW,CW}, defined as follows. We define maps $\chi_0\colon U \times \Vzeroeps \rightarrow \SLtwoC$, $\chi_{\infty}\colon U \times \Vinfeps \rightarrow \SLtwoC$ by the relations
\begin{align*}
\chi_0(x, z) &:= \psi_0(x)^{-1} \cdot \Psi_0(x, z), & &(x, z) \in U \times \Vzeroeps
\\
\chi_{\infty}(x, z) &:= \psi_{\infty}(x)^{-1} \cdot \Psi_{\infty}(x, z),  & &(x, z) \in U \times \Vinfeps.
\end{align*}
$\chi_0$ is analytic in $z$ for all $z \in V_0$, with $\chi(x, 0) = \Id$, for all $x \in U$, and is a solution of the system
\begin{align*}
\left[ \left( \partial_{\vbar} - J_{\vbar} J^{-1} \right) + z \partial_u \right] \chi_0(x, z) &= 0,
\\
\left[ \left( \partial_{\ubar} - J_{\ubar} J^{-1} \right) - z \partial_v \right] \chi_0(x, z) &= 0,
\end{align*}
for all $(x, z) \in U \times \Vzeroeps$. Similarly, $\chi_{\infty}$ is analytic in $z$ for all $z \in \Vinfeps$, with $\chi_{\infty}(x, \infty) = \Id$, for all $x \in U$. Note that we have
\[
\chi_{\infty}(x, \lambda) = \left( \chi_0(x, \sigma(\lambda)) \right)^{-\dagger}, \mbox{ for all $(x, \lambda) \in U \times \Vinfeps$.}.
\]

\vs
Based on the work of~\cite{CGSW,CGW,CW}, we have the following result from~\cite{Cr}:

\begin{proposition}
\label{prop:Jsymmetry}
Let $T\colon U \times \Veps \rightarrow \sltwoC$ obey the relations
\[
\left( \partial_{\vbar} + \lambda \partial_u \right) T(x, \lambda) =
\left( \partial_{\ubar} - \lambda \partial_v \right) T(x, \lambda) = 0,
\]
and be analytic in $\lambda$ on a neighbourhood, {\Veps}, of the unit circle $S^1 \subset \mathbb{C}$. Then
\begin{align}
\dot{J}(x, s) &= \chi_{\infty}(x, \lambda) T(x, \lambda) \chi_{\infty}(x, \lambda)^{-1} \cdot J
\nonumber\\
&\hskip3cm + J \cdot \chi_0(x, \sigma(\lambda)) T(x, \lambda)^{\dagger} \chi_0(x, \sigma(\lambda))^{-1}
\nonumber
\\
&= \psi_{\infty}(x)^{-1}
\left[ \Psi_{\infty}(x, \lambda) T(x, \lambda) \Psi_{\infty}(x, \lambda)^{-1} \right.
\nonumber\\
&\hskip3cm \left.+ \Psi_0(x, \sigma(\lambda)) T(x, \lambda)^{\dagger} \Psi_0(x, \sigma(\lambda))^{-1} \right] \psi_0(x)
\label{Jsymmetry}
\end{align}
is a solution of the linearisation equation~\eqref{linearisation}, for all $x \in U$, and all $\lambda \in \Veps$.
\end{proposition}

\begin{remark}
In the case where the function $T$ is independent of $(u, v)$, it defines an element of the loop group $\Lambda \SLtwoC$ with a holomorphic extension to an open neighbourhood of $S^1$ in $\mathbb{C}^*$. The algebra of symmetries generated by such $T$ is then isomorphic to the Kac-Moody algebra of {\sltwoC}~\cite{CGSW,CGW,CW}.
\end{remark}

The natural question is how to exponentiate the above algebra into a group action on the space of solutions of the {\sdyme}. A solution to this problem is given by the following result
\begin{theorem}\cite[Chapter~IV.C]{Cr}
\label{thm:Crane}
Let $g\colon X \times S^1 \rightarrow \SLtwoC$ be a smooth map that admits a continuous extension to a holomorphic map $g\colon X \times \Veps \subset \cpthree \rightarrow \SLtwoC$, for some $\epsilon > 0$. Then the action of $g$ on the patching matrix, $G(x, z)$, is defined by
\begin{equation}
G(x, z) \mapsto \left( g \cdot G \right) (x, z) := g(x, z) \cdot G(x, z) \cdot g^*(x, z).
\label{crane}
\end{equation}
This equation defines an action of the set of such maps $g$ on the space of self-dual connections on $X$.  If $g$ extends holomorphically to $z \in \Vzeroeps$, then the corresponding transformation is a holomorphic change of basis on the bundle over $\pi^{-1}(X)$, which leaves the self-dual connection, $\mathbf{A}$, unchanged.
\end{theorem}

\begin{remark}
The above group action on the solution space have been given a cohomological description by Park (see~\cite{Pa} and references therein), which has been further investigated by Popov and Ivanova (see~\cite{Iv2,Iv1,Po} and references therein).
\end{remark}

\begin{remark}
The group action~\eqref{crane} is slightly unusual since, in integrable systems theory, it is usually adjoint or coadjoint orbits of Lie groups that turn out to be relevant. If we consider the case where $G$ and $g$ are constant, and study the action of {\SLtwoC} on itself by $(g, G) \mapsto g \cdot G := g G g^{\dagger}$, then the generic orbits of this action are five-dimensional. As such, the orbits do not carry the invariant symplectic structures that one would associate with, for example, coadjoint orbits. A brief investigation of this orbit structure is given in Appendix~\ref{orbits}.
\end{remark}

The connection between Theorem~\ref{thm:Crane} and the transformation~\eqref{Jsymmetry} is given by the following:
\begin{theorem}
Given $T\colon U \times \Veps \rightarrow \sltwoC$ as in Proposition~\ref{prop:Jsymmetry}, the corresponding flow on the space of patching matrices is given by
\begin{equation}
\dot{G}(x, z) = - T(x, z) G(x, z) - G(x, z) T^*(x, z) + \rho_{\infty}(x, z) G(x, z) + G(x, z) \rho_0(x, z)
\label{Gdot}
\end{equation}
for $(x, z) \in \mathbb{R}^4 \times \Veps$. In this equation, $\rho_0\colon \mathbb{R}^4 \times \Vzeroeps \rightarrow \sltwoC$ and $\rho_{\infty}\colon \mathbb{R}^4 \times \Vinfeps \rightarrow \sltwoC$ are analytic functions of $z$ on the respective regions and satisfy
\begin{align*}
\left( \partial_{\vbar} + z \partial_u \right) \rho_0(x, z) =
\left( \partial_{\ubar} - z \partial_v \right) \rho_0(x, z) = 0,
\\
\left( \partial_{\vbar} + z \partial_u \right) \rho_{\infty}(x, z) =
\left( \partial_{\ubar} - z \partial_v \right) \rho_{\infty}(x, z) = 0.
\end{align*}
\end{theorem}

\begin{remark}
It follows from a similar argument to that given in the proof of Proposition~1 (b) of~\cite{Cr} that the terms $\rho_0$ and $h_{\infty}$ in the above formula may be absorbed into a change of holomorphic frame on the sets $z \in \Vzeroeps$ and $\Vinfeps$, respectively.
\end{remark}

\begin{remark}
The group action~\eqref{crane} was derived in~\cite{Cr}, arguing by analogy with the action of dressing transformations in harmonic map theory. It has been investigated further in, for example, \cite{Iv2,Iv1,Po}. The first direct derivation of the infinitesimal result~\eqref{Gdot} from the generator~\eqref{Jsymmetry}, to my knowledge, appears in~\cite{sdym}.
\end{remark}

\section{Reducible connections}
\label{sec:reducible}

Recall~\cite[Chapter~3]{FU} that a connection, $D$, on an {\SUtwo} bundle $\pi\colon E \rightarrow U$ is reducible if the group of gauge transformations $\mathcal{G} := C^{\infty}(U, \SUtwo)$ modulo its centre does not act freely on the connection $D$. We now proceed to derive the most general form of a reducible connection on a simply-connected, open subset of $\mathbb{R}^4$. In doing so, we make extensive use of the classical Pauli matrices, which we define as follows:
\[
\boldsymbol{\tau} \equiv \left( \tau_1, \tau_2, \tau_3 \right) =
\left(
\begin{pmatrix} 0 &1 \\ 1 &0 \end{pmatrix},
\begin{pmatrix} 0 &-i \\ i &0 \end{pmatrix},
\begin{pmatrix} 1 &0 \\ 0 &-1 \end{pmatrix}
\right).
\]

\begin{proposition}
Let $U$ be a simply-connected, open subset of $\mathbb{R}^4$, $a \in C^{\infty}(U, \mathbb{R})$ a harmonic function. We define the connection
\begin{equation}
\mathbf{A} = \frac{1}{2} \left( \partial a - \overline{\partial} a \right) \tau_3 \in \Omega^1(U, \sutwo).
\label{reducibleconnection}
\end{equation}
Then the connection, $\mathbf{A}$, is reducible and satisfies the {\sdyme} on $U$. Conversely, up to gauge transformation, all reducible self-dual connections on $U$ arise in this way.
\label{prop:redconn}
\end{proposition}

\begin{proof}
A connection is reducible if and only if there exists a parallel section, $\eta$, of the adjoint bundle $\mathrm{ad}\,\sutwo$~\cite[Theorem~3.1]{FU}. In terms of local coordinates $(u, v) \in \mathbb{C}^2 \cong \mathbb{R}^4$, and relative to a local trivialisation of the bundle $E$, this condition implies that
\[
\frac{\partial}{\partial x^a} \eta + \left[ A_a, \eta \right] = 0.
\]
It follows from equations~\eqref{Acomponents} that, for a parallel section $\eta$, the map $A\colon \mathbb{C}^2 \rightarrow \sltwoC$ defined by the equation $\eta = \psi_0(x) A(u, v) \psi_0(x)^{-1}$ is holomorphic. In a similar fashion, using ~\eqref{Acomponents} and the relationship between $\psi_0$ and $\psi_{\infty}$, one may verify that $\eta = - \psi_{\infty}(x) A(u, v)^{\dagger} \psi_{\infty}(x)^{-1}$.
%This equation, along with equations~\eqref{Acomponents},
%implies that there exists a holomorphic map $A\colon \mathbb{C}^2 \rightarrow \sltwoC$ with
%\[
%\eta = \psi_0(x) A(u, v) \psi_0(x)^{-1} = - \psi_{\infty}(x) A(u, v)^{\dagger} \psi_{\infty}(x)^{-1}.
%\]
These relations imply that
\begin{equation}
J A(u, v) + A(u, v)^{\dagger} J = 0.
\label{JA}
\end{equation}
Note that this equation and the fact that $\det J \neq 0$ implies that $\det A(u, v) = \det A(u, v)^*$. Therefore $\det A(u, v)$ is \emph{real}. Since $A$ is holomorphic in $u$ and $v$, it follows that $\det A$ is a real constant.

Let $A = \left( \mathbf{R} + i \mathbf{I} \right) \cdot \boldsymbol{\tau}$, with $\mathbf{R}, \mathbf{I}\colon U \rightarrow \mathbb{R}^3$. The fact that $A$ depends holomorphically on $(u, v)$ implies that
\begin{equation}
\partial_t \mathbf{R} = \partial_x \mathbf{I}, \qquad
\partial_x \mathbf{R} = - \partial_t \mathbf{I}, \qquad
\partial_y \mathbf{R} = - \partial_z \mathbf{I}, \qquad
\partial_z \mathbf{R} = \partial_y \mathbf{I}.
\label{CR}
\end{equation}
Moreover, we find that
\[
\det A = |\mathbf{I}|^2 - |\mathbf{R}|^2 - 2 i \langle \mathbf{R}, \mathbf{I} \rangle.
\]
Since $\det A$ is real, we deduce that $\langle \mathbf{R}, \mathbf{I} \rangle = 0$. We now let $J = \Lambda + \boldsymbol{\lambda} \cdot \boldsymbol{\tau}$, with $\boldsymbol{\lambda}\colon U \rightarrow \mathbb{R}^3$ and the condition that $\det J = 1$ implies that we require $\Lambda^2 = 1 + |\boldsymbol{\lambda}|^2$. Imposing~\eqref{JA}, we find that we require
\begin{equation}
\Lambda \mathbf{R} = \boldsymbol{\lambda} \times \mathbf{I}.
\label{RI}
\end{equation}
This equation implies that
\[
\Lambda^2 |\mathbf{R}|^2 = |\boldsymbol{\lambda} \times \mathbf{I}|^2 = |\boldsymbol{\lambda}|^2 |\mathbf{I}|^2 - \langle \mathbf{I}, \boldsymbol{\lambda} \rangle^2.
\]
Therefore
\begin{align*}
\det A = |\mathbf{I}|^2 - |\mathbf{R}|^2 &=
\frac{1}{|\boldsymbol{\lambda}|^2} \left( \Lambda^2 |\mathbf{R}|^2 + \langle \mathbf{I}, \boldsymbol{\lambda} \rangle^2 \right) - |\mathbf{R}|^2
\\
&= \frac{1}{|\boldsymbol{\lambda}|^2} \left( |\mathbf{R}|^2 + \langle \mathbf{I}, \boldsymbol{\lambda} \rangle^2 \right)
\\
&\ge 0
\end{align*}
For $\boldsymbol{\lambda} \neq 0$, equality occurs in this inequality if and only if $\mathbf{R} = 0$ and $\langle \mathbf{I}, \boldsymbol{\lambda} \rangle = 0$. From~\eqref{RI}, it follows that, in this case, $\mathbf{R} = \mathbf{I} = 0$, so $A = 0$. Moreover, if $\boldsymbol{\lambda} = 0$, then~\eqref{JA} implies that $\mathbf{R} = 0$, so $\det A = |\mathbf{I}|^2$, which is, again, strictly positive unless $\mathbf{I} = 0$ and, hence, $A = 0$.

\vs
To summarise, the fact that $\det A$ is real, along with~\eqref{JA}, implies that $\det A$ is a non-negative constant. Moreover, $\det A = 0$ if and only if $A = 0$. Since we are assuming $A \neq 0$ and since any constant multiple of a parallel section is also parallel we may, without loss of generality, assume that $\det A = 1$. In this case, it follows that the eigenvalues of $A$ are $\pm i$.

\vs
Note that we still have the freedom to rotate the $\psi$'s, as given in~\eqref{psifreedom}. It follows that $A$ transforms under the adjoint action of $R^{-1}$:
\begin{equation}
A(u, v) \mapsto {\widetilde A}(u, v) := R(u, v)^{-1} A(u, v) R(u, v) = \mathrm{Ad}_{R^{-1}} A.
\label{Atransformation}
\end{equation}
We now write $A$ in the form
\[
A(u, v) = \begin{pmatrix} a(u, v) &b(u, v) \\ c(u, v) &-a(u, v) \end{pmatrix}
\]
where $a, b, c$ are holomorphic functions of $(u, v)$. On a neighbourhood of any point $(u, v)$ at which $a(u, v) \neq -i$, the holomorphic change of frame
\[
R_+(u, v) = \begin{pmatrix} a+i &b \\ c &a+i \end{pmatrix}
\]
has the property that
\[
R_+(u, v)^{-1} A(u, v) R_+(u, v) = i \tau_3.
\]
Similarly, for $a(u, v) \neq +i$,
\[
R_-(u, v) = \begin{pmatrix} -b &a-i \\ a-i &c \end{pmatrix}
\]
gives a holomorphic change of frame with the property that
\[
R_-(u, v)^{-1} A(u, v) R_-(u, v) = i \tau_3.
\]
As such, given any point $p \in X$, there exists a neighbourhood of $p$ and a holomorphic frame such that $A = i \tau_3$ in that frame.

From~\eqref{JA}, it follows that there exist real functions $\alpha, \beta$ such that $J = \alpha \, \Id + \beta \tau_3$. Since $\det J = 1$, we have $\alpha^2 = 1 + \beta^2$. Since $J$ is continuous, $\alpha$ will have constant sign, so we assume that $\alpha > 0$. Therefore, since $U$ is assumed simply-connection, we may consistently define a real-valued function $a$ with the property that $\alpha = \cosh a, \beta = \sinh a$. It then follows that
\begin{equation}
J = \exp \left( a \tau_3 \right).
\label{Jexp}
\end{equation}
We therefore have
\[
J^{-1} J_u = a_u \, \tau_3, \qquad J^{-1} J_v = a_v \, \tau_3.
\]
Imposing the Yang--Pohlmeyer equation implies that $a$ is harmonic:
\[
\left( \partial_u \partial_{\ubar} + \partial_v \partial_{\vbar} \right) a = 0.
\]

From~\eqref{Jexp}, we see that, up to a gauge transformation, we may take
\[
\psi_0 = \exp \left( \frac{a}{2} \tau_3 \right), \qquad \psi_{\infty} = \exp \left( - \frac{a}{2} \tau_3 \right).
\]
The form of the connection given in equation~\eqref{reducibleconnection} then follows from equation~\eqref{Acomponents}. The parallel section, $\eta$, is equal to $i \tau_3$.
\end{proof}

\begin{example}
The case
\[
J(x) = \exp \left\{\left( |u|^2 - |v|^2 \right) \tau_3 \right\}
\]
corresponds to $a = |u|^2 - |v|^2$ and, therefore, defines a reducible connection. In this case, the connection is non-singular on $\mathbb{R}^4$. However, the curvature is not $L^2$, and therefore the connection cannot be extended to $S^4$~\cite{U}.

In this particular case the connection is algebraically special, in the sense that, in addition to satisfying the {\sdyme}~\eqref{sdym}, the curvature satisfies
\[
F_{u{\vbar}} = 0, \qquad F_{v{\ubar}} = 0.
\]
It can be shown that all algebraically special connections arise in this way, and are thus reducible.
\end{example}

Recall~\cite{HitchinLinear} that there is a $1-1$ correspondence between harmonic functions on $U \subseteq \mathbb{R}^4$ and sheaf cohomology classes in $H^1(\widehat{U}, \mathcal{O}(-2))$, where $\widehat{U} := \pi^{-1}(U) \subseteq \cpthree$. In the present case, such a cohomology class may be represented by a holomorphic function\footnote{By holomorphic, we mean with respect to the complex structure induced on $U \times \mathbb{C}^*$ as a subset of {\cpthree}.} $f\colon U \times \mathbb{C}^* \rightarrow \mathbb{C}$. In terms of homogeneous coordinates on {\cpthree}, we have $f(\lambda \mathbf{z}) = \lambda^{-2} f(\mathbf{z})$. The corresponding harmonic function on $U \subseteq \mathbb{R}^4$ is then given by the contour integral
\begin{equation}
a(x) = \frac{1}{2 \pi i} \oint_{\gamma} f(u - w \vbar, v + w \ubar, w) dw,
\label{aeqn}
\end{equation}
where $\gamma := \{ w \in \mathbb{C} \subset \mathbb{C}P^1 : |w| = 1 \}$.

\begin{proposition}
Given the connection as in Proposition~\ref{prop:redconn}, the patching matrix for the holomorphic bundle on $\widehat{U}$ may be taken as
\[
G(x, z) = \exp \left[ \frac{1}{2} \left( F(x, z) + F^*(x, z) \right) \tau_3 \right]
\]
where
\begin{equation}
F(x, z) := \frac{1}{2 \pi i} \oint_{\gamma} \frac{w+z}{w-z} f(u - w \vbar, v + w \ubar, w) dw,
\label{F}
\end{equation}
and the holomorphic function $f$ is a representative of the cohomology class in $H^1(\widehat{U}, \mathcal{O}(-2))$ corresponding to the harmonic function $a$.
%This representation of $G$ is unique up to holomorphic changes of frame on $X \times U_0$ and $X \times U_{\infty}$.
\label{Gred}
\end{proposition}
\begin{proof}
We assume that there exists $\Psi_0(x, z)$ of the form $\exp \left( \frac{1}{2} F(x, z) \tau_3 \right)$, with $F(x, 0) = a(x)$. From~\eqref{ALP} for $\Psi$ and the explicit form of the connection, we deduce that $F$ must be analytic in $z$ and satisfy the relations
\[
\left( \partial_{\ubar} - z \partial_v \right) F(x, z) = \left( \partial_{\ubar} + z \partial_v \right) a(x), \qquad
\left( \partial_{\vbar} + z \partial_u \right) F(x, z) = \left( \partial_{\vbar} - z \partial_u \right) a(x).
\]
From~\eqref{aeqn} we then calculate
\begin{align*}
\left( \partial_{\ubar} + z \partial_v \right) a(x) &=
\left( \partial_{\ubar} + z \partial_v \right) \frac{1}{2 \pi i} \oint_{\gamma} f(u - w \vbar, v + w \ubar, w) dw
\\
&= \frac{1}{2 \pi i} \oint_{\gamma} \left( w + z \right) \left( \partial_2 f \right) (u - w \vbar, v + w \ubar, w) dw
\\
&= \frac{1}{2 \pi i} \oint_{\gamma} \frac{w + z}{w - z} \left( w - z \right) \left( \partial_2 f \right) (u - w \vbar, v + w \ubar, w) dw
\\
&= \frac{1}{2 \pi i} \oint_{\gamma} \frac{w + z}{w - z} \left( \partial_{\ubar} - z \partial_v \right) f(u - w \vbar, v + w \ubar, w) dw
\\
&= \left( \partial_{\ubar} - z \partial_v \right) \frac{1}{2 \pi i} \oint_{\gamma} \frac{w + z}{w - z} f(u - w \vbar, v + w \ubar, w) dw.
\end{align*}
Along with a similar calculation for $\left( \partial_{\vbar} - z \partial_u \right) a(x)$, we have the above candidate for $F$ given in~\eqref{F}. By its definition as a contour integral, it follows that $F$ is analytic in $z$, for $z$ in the interior of $\gamma$. Since the function $\frac{w + z}{w - z} f(u - w \vbar, v + w \ubar, w)$ is continuous for $w \in \gamma$, and $\gamma$ is compact, the Dominated Convergence Theorem implies that $\lim_{z \rightarrow 0} F(x, z) = a(x)$. As such, $F$ has all of the required properties.
\end{proof}

\begin{remark}
\label{rem:Gred}
Note that the patching matrix in Proposition~\ref{Gred} takes the form $G(x, z) = \exp \left( \varphi(x, z) \tau_3 \right)$, where $\varphi$ is a holomorphic function of $(u-z\vbar, v+z\ubar, z)$ that satisfies $\varphi^*(x, z) = \varphi(x, z)$.
\end{remark}

\begin{example}
In the case of our algebraically special connection, where $a = |u|^2 - |v|^2$, we may take $f(x, z) = \frac{1}{z^2} (u-z\vbar)(v+z\ubar)$. We then find that $F(x, z) = |u|^2 - |v|^2 - 2 \ubar \vbar$ and $G(x, z) = \exp \left[ \frac{1}{z} \left( u - z \vbar \right) \left( v + z \ubar \right) \tau_3 \right]$.
\end{example}

\section{Orbit of the flat connection}
\label{flatorbit}

Let $T\colon U \times S^1$ be analytic on a neighbourhood of $U \times S^1$ in $U \times \mathbf{C}^* \subset \cpthree$, with the property that
\begin{equation}
\left[ T, T^* \right] = 0.
\label{TTstar}
\end{equation}
Given a solution of the {\sdyme}, described by patching matrix $G(x, z)$, we may consider the one-parameter family of connections generated by $T$ via the flow~\eqref{Gdot}. Since the functions $\rho_0$, $\rho_{\infty}$ can be removed by a holomorphic change of basis on the regions $U \times \Vzeroeps$ and $U \times \Vinfeps$ we may, without loss of generality, fix the holomorphic bases by setting $\rho_0 = \rho_{\infty} = 0$. The unique solution of~\eqref{Gdot} with initial conditions $G(x, z)$ is then
\[
G_t(x, z) = \exp \left( - t T(x, z) \right) \, G(x, z) \, \exp \left( - t T(x, z)^* \right).
\]
In general, one would perform a Birkhoff splitting of $G_t$, which would yield connections $\mathbf{A}_t$ generated from the connection, $\mathbf{A}$, corresponding to the patching matrix $G$. (Generally, there will be jumping points at which $G$ does not admit a splitting of the form~\eqref{patching}, so we will need to shrink the set $U$ accordingly. The set of such points will, generically, be of strictly positive codimension in $U$.)

A case of particular interest to us is when the initial connection is flat, in which case we may take $G(x, z) = 1$. We then have
\[
G_t(x, z) = \exp \left( - t \left( T(x, z) + T(x, z)^* \right) \right)
\]
as the patching matrix of connections that lie in this orbit of the flat connection. As a special case of this construction, letting $T = - \frac{1}{2} \varphi(x, z) \tau_3$ where $\varphi\colon U \times \Veps \rightarrow \SLtwoC$ satisfies
\[
\left( \partial_{\ubar} - z \partial_v \right) \varphi =
\left( \partial_{\vbar} + z \partial_u \right) \varphi = 0
\]
and is analytic in $z \in \mathcal{V}_{\epsilon}$ for some $\epsilon > 0$. Then, assuming that $\epsilon$ is chosen sufficiently small that $G(x, z)$ is analytic on $U \times \Veps$, we deduce that
\[
G_t(x, z) = \exp \left( t \varphi(x, z) \tau_3 \right), \qquad (x, z) \in U \times \Veps.
\]

\vs
In light of this construction, and the classification of patching matrices arising from reducible connections in the previous section, we deduce:
\begin{theorem}
Let $\mathbf{A}$ be a reducible connection on an open subset $U \subset \mathbb{R}^4$. Then $\mathbf{A}$ lies on the orbit of the flat connection on $U$ under the action of the non-local symmetry group of the {\sdyme}.
\end{theorem}

\begin{remark}
As mentioned earlier, the set $U$ on which the self-dual connection is defined will generally shrink under the action of the symmetry group. In the case of reducible connections, where one may start from the flat connection on $\mathbb{R}^4$, then there do exist reducible connections defined on the whole of $\mathbb{R}^4$. (Our algebraically special connection is an example of such.) Since there are no non-trivial reducible connections on $S^4$, however, Uhlenbeck's theorem~\cite{U} implies that the curvature of such a connection cannot be $L^2$. (This property may also be checked directly from the explicit form of the connection.)
\end{remark}

\begin{remark}
Takasaki~\cite{T} has argued that, if we drop the reality conditions on self-dual connections and consider $\SLtwoC$ connections rather than $\SUtwo$ ones, then the group action generated by transformations of the form
\begin{equation}
\dot{J}(x, s) = \chi_{\infty}(x, \lambda) T(x, \lambda) \chi_{\infty}(x, \lambda)^{-1} \cdot J
\label{Takasaki}
\end{equation}
is transitive on the space of local solutions of the $\SLtwoC$ {\sdyme}. In the current context, we are explicitly  restricting ourselves (via the form of equation~\eqref{Jsymmetry}) to transformations that preserve the $\SUtwo$ nature of the connection, in which case there is no reason to believe that the group action should be transitive. In an analogous situation in the theory of harmonic maps into Lie groups, one can show that transformations analogous to~\eqref{Takasaki} map real extended harmonic maps to real harmonic maps if and only if the action is trivial~\cite[Proposition~3.4]{BG} (i.e. $g \cdot \Phi = \Phi$). It is, similarly, expected that transformations of the form~\eqref{Takasaki} will map {\SUtwo} connections to {\SUtwo} connections if and only if the connections coincide.
\end{remark}

\section{Harmonic maps of finite type}
\label{sec:HMFT}

It is clear from the discussion in the previous section that the reducible connections on a simply-connected, open subset $U \subseteq \mathbb{R}^4$ are a special case of a more general type of connection. In particular, given a map $T\colon U \times \Veps \rightarrow \SLtwoC$ satisfying the commutator condition~\eqref{TTstar}, then the patching matrix
\begin{equation}
G(x, z) = \exp \left( T(x, z) + T(x, z)^* \right)
\label{GFT}
\end{equation}
will generate a solution of the {\sdyme} on a subset of $U$.

The forms of condition~\eqref{TTstar} and the patching matrix~\eqref{GFT} are reminiscent of formulae that appear when one considers harmonic maps of finite type into Lie groups (see, e.g., \cite[Chapter~24]{Gu} and~\cite{BP1,BP2} for harmonic maps into $k$-symmetric spaces). Recall that, in this context, we consider Lie groups $\mathcal{G}$, $\mathcal{G}_1$, $\mathcal{G}_2$ such that $\mathcal{G} = \mathcal{G}_1 \cdot \mathcal{G}_2$ (in the sense that, given $g \in \mathcal{G}$, there exist unique $g_1 \in \mathcal{G}_1$, $g_2 \in \mathcal{G}_2$ such that $g = g_1 g_2$). At the Lie algebra level, we have a direct sum decomposition $\mathfrak{g} = \mathfrak{g}_1 + \mathfrak{g}_2$, and we denote the projections onto the two summands by $\pi_1$, $\pi_2$. In the case where $Ad_{\mathcal{G}_1} \mathfrak{g}_2 \subseteq \mathfrak{g}_2$, then various Lax flows on $\mathfrak{g}$ can be solved explicitly. In particular, let $\mathbf{J}_1$, $\mathbf{J}_2$ be invariant vector fields on $\mathfrak{g}$\footnote{i.e. $\mathbf{J}_1, \mathbf{J}_2\colon \mathfrak{g} \rightarrow \mathfrak{g}$ satisfy $\mathbf{J}_1(Ad_g \mathbf{v}) = Ad_g \mathbf{J}_1(\mathbf{v})$ for all $g \in \mathcal{G}$ and $\mathbf{v} \in \mathfrak{g}$ and similarly for $\mathbf{J}_2$} consider the Lax equations
\begin{subequations}
\begin{align}
\partial_s X(s, t) = \left[ X(s, t), \left( \pi_1 \circ \mathbf{J}_1 \right)(X(s, t)) \right],
\\
\partial_t X(s, t) = \left[ X(s, t), \left( \pi_1 \circ \mathbf{J}_2 \right)(X(s, t)) \right],
\end{align}\label{Lax}\end{subequations}
for a map $X\colon \mathbb{R}^2 \mapsto \mathfrak{g}$ with initial conditions
\[
X(0, 0) = \mathbf{V} \in \mathfrak{g}.
\]
These equations are compatible, and the solution to this problem may be written in the form
\[
X(s, t) = Ad_{F(s, t)^{-1}} \mathbf{V}
\]
where $F\colon \mathbb{R}^2 \rightarrow \mathcal{G}_1$ takes the form
\[
F(s, t) = \exp \left( s \left( \pi_1 \circ \mathbf{J}_1 \right)(\mathbf{V}) + t \left( \pi_1 \circ \mathbf{J}_2 \right)(\mathbf{V}) \right), \qquad
(s, t) \in \mathbb{R}^2.
\]
The connection with harmonic maps arises if we let $G$ be a compact Lie group and use the standard loop-group decompositions (see, e.g.,~\cite[Chapter~12]{Gu} and~\cite[Chapter~8]{PS}) to take $\mathcal{G} := \Lambda G^{\mathbb{C}}$, $\mathcal{G}_1 := \Omega G$, $\mathcal{G}_2 := \Lambda_+ G^{\mathbb{C}}$. If we now impose that the initial conditions for the Lax equation~\eqref{Lax} corresponds to an element of the loop group of $G$ (rather than $G^{\mathbb{C}}$) and that its Laurent expansion of which lies between degrees $-d$ and $d$:
\[
\mathbf{V} = \left[ \lambda \mapsto f(\lambda) \equiv \sum_{n=-d}^d \alpha_n \lambda^n \right] \in \Lambda G^{\mathbb{C}},
\]
then it turns out that the map $X$ also has Laurent expansion which lies between degrees $-d$ and $d$. Moreover, $F\colon \mathbb{R}^2 \rightarrow \Omega G$ is automatically the extended solution corresponding to a harmonic map $\varphi\colon \mathbb{R}^2 \rightarrow G$. (In particular, $\varphi(s, t) = \left. F(s, t) \right|_{\lambda = -1}$.) Such harmonic maps are of \emph{finite type}.

\vs
In the context of self-dual Yang--Mills connections, the analogue of harmonic maps of finite type for self-dual Yang--Mills fields would appear to be patching matrices of the form
\[
G(x, z) = \exp \Phi(x, z),
\]
where $\Phi\colon U \times \Veps \rightarrow \SLtwoC$ satisfies
\begin{itemize}
\item[1). ]
$\left( \partial_{\ubar} - z \partial_v \right) \Phi(x, z) = \left( \partial_{\vbar} + z \partial_u \right) \Phi(x, z) = 0$;
\item[2). ]
$\Phi^*(x, z) = \Phi(x, z)$
\item[3). ]
$\Phi(x, z)$ is analytic in $z$ for $z \in \Veps$ for some $\epsilon > 0$, and there exists $d \in \mathbb{N}_0$ such that $\Phi$ has a finite Laurent expansion of the form
\[
\Phi(x, z) = \sum_{n=-d}^d a_n(x) z^n, \qquad (x, z) \in U \times \Veps,
\]
for some $a_i\colon U \rightarrow \mathfrak{g}^{\mathbb{C}}$, $i = -d, \dots, d$ on this set.
\end{itemize}
In particular, fixing a point $p \in U$ then the finite Laurent expansion at $p$ is analogous to the initial condition $\mathbf{V}$ having finite Laurent expansion. Moreover, Condition~1) above is then the analogue of the Lax equations~\eqref{Lax} satisfied by the map $X$ in the harmonic map case.

\begin{definition}
We will call a solution of the {\sdyme} for which there exists a patching matrix that satisfies the above criteria a \emph{self-dual connection of finite type}\footnote{Or, of \emph{type $d$}, when we wish to be more specific}.
\end{definition}

\begin{remark}
The conditions above imply that the maps $a_i$ satisfy the conditions
\begin{subequations}
\begin{align}
\partial_{\ubar} a_{-d} &= 0 & &\partial_{\ubar} a_{-d} = 0, & &
\\
\partial_{\ubar} a_{n+1} &= \partial_v a_n,
& &\partial_{\vbar} a_{n+1} = - \partial_u a_n, & & n = -d, \dots, d-1
\\
\partial_v a_d &= 0 & & \partial_u a_d = 0. & &
\end{align}\label{SDFT}\end{subequations}
For $d=0$, we deduce that $G$ is constant, and therefore the corresponding self-dual connection $\mathbf{A}$ is flat. For $d \ge 1$, the algebraic condition~\eqref{TTstar} imposes non-trivial restrictions on the coefficients $a_i$. Note that our algebraically special connection is a self-dual connection of type $1$.
\end{remark}

\begin{remark}
Let $\mathbf{A}$ be a reducible connection defined by a harmonic function $a$ as in~\eqref{reducibleconnection}. Letting $a_0(x) := a(x)$ then $\mathbf{A}$ is a self-dual connection of finite type if and only if there exists $d>0$ and functions $a_{-d}, \dots, a_d$ such that equations~\eqref{SDFT} hold. In particular, this condition implies that
\[
\frac{\partial^d a}{\partial \ubar^r \partial \vbar^s} = 0, \mbox{ for all $r, s$ such that $r + s = d$}.
\]
Therefore, $\Phi$ is necessarily a polynomial of degree less than or equal to $d$ in $(u, \ubar, v, \vbar)$. As such, the space of self-dual connections of type $d$ is necessarily finite-dimensional.
\end{remark}

\begin{remark}
The most restrictive case is when we impose that $G$ splits in the form~\eqref{patching} with
\begin{align*}
\Psi_0(x, z) &= \exp \left( \frac{a_0(x)}{2} + \sum_{n=1}^d a_n(x) z^n \right),
\\
\Psi_{\infty}(x, z) &= \exp \left( - \frac{a_0(x)}{2} -  \sum_{n=-d}^{-1} a_n(x) z^n \right).
\end{align*}
Such a splitting only occurs if $\left[ a_i(x), a_j(x) \right] = 0$, for all $i, j = -d, \dots, d$. Since {\SLtwoC} is of rank one, it follows that there exists a constant element $\alpha \in \SLtwoC$ such that $a_i(x) = \varphi_i(x) \alpha$, for functions $\varphi_i\colon U \mapsto \mathbb{C}$. A change of basis (rotating so that $\alpha \mapsto \tau_3$) implies that such patching matrices give rise to reducible connections when $a_0$ is real.
\end{remark}

\begin{remark}
One of the main differences between the integrable systems approach to harmonic maps and the {\sdyme} is the form of the symmetry group action on the solutions. In the case of harmonic map equations from a domain $X \subseteq \mathbb{R}^2$ to a Lie group $G$, one interprets the harmonic map equations as implying the existence of a holomorphic map $E\colon X \rightarrow \Omega G$ into the based loop group of $G$. The \lq\lq dressing action\rq\rq\ on the space of harmonic maps is then induced by the action of various groups on the group $\Omega G$~\cite{GO,UhlenbeckJDG}. In particular, the symmetry group acts only on the space where the map $E$ takes its values, rather than on the map $E$ itself. In the case of the {\sdyme}, the object of study is the patching matrix $G\colon U \times \Veps \rightarrow \sltwoC$, and the group action~\eqref{crane} acts non-trivially on the map $G$. This difference is the main issue that makes the case of {\sdyme} more complicated. As remarked earlier, the particular form of the group action~\eqref{crane} implies that many of the techniques used to study orbits in the harmonic map case have no direct analogue in the self-dual Yang--Mills case.
\end{remark}

\section{Final remarks}
Our main result is that reducible connections that satisfy the {\sdyme} on simply-connected, open subsets of $\mathbb{R}^4$ lie in the orbit of the flat connection under the action of the non-local symmetry group of these equations found in~\cite{CGSW,CGW,CW}. In particular, such connections lie within a larger class of solutions, discussed in Section~\ref{flatorbit}, defined by a holomorphic function $T(x, z)$ with the property that $\left[ T(x, z), T^*(x, z) \right] = 0$. This condition defines a class of solutions of the {\sdyme} that seem quite natural from the integrable systems point of view, and suggests a connection with the theory of harmonic maps of finite type. Whether the analogy with such harmonic maps may be extended, and techniques developed in, for example~\cite{BP1,BP2}, may be adapted to the study of our class of self-dual connections, is under investigation.

It is clear that the work here (and in the sister paper~\cite{sdym}) may be extended in several ways. The investigation of the symmetry group on the one-instanton moduli space on the four-manifold {\cptwo} would be of particular interest, since, in this case, the standard reducible connection is also $L^2$ so we have reducible and irreducible connections in the same moduli space. Such an investigation would yield further information concerning the different behaviour of reducible connections studied here and the instanton connections studied in~\cite{sdym} under the symmetry group. In a different direction, given that the reducible connections and the class discussed in Section~\ref{flatorbit} seem quite a natural family of solutions to investigate from the point of view of integrable systems, it would be of interest to investigate whether there are similar families of self-dual Ricci-flat four-manifolds (for example, those with algebraically special self-dual Weyl tensor) that arise naturally from the symmetries of, for example, Plebanski's equations~\cite{JFP}.

We should also point out that we have exclusively considered the {\sdyme} on Riemannian manifolds, due to the original motivation of Donaldson theory. It is more usual to investigate the integrable systems aspects of the {\sdyme} equations on manifolds of signature $(--++)$ (see, e.g., \cite{MW} for an extensive treatment of this topic). It would be of interest to investigate the action of symmetries on, for example, reducible connections in the case of signature $(--++)$.

\vs
Viewed in conjunction with the results of the companion paper~\cite{sdym}, where instanton solutions of the {\sdyme} were investigated, it appears that the action of the non-local symmetry group on the space of solutions of the {\sdyme} is quite different in the two cases. In the case of instanton moduli spaces, evidence was found that the orbits of the symmetry group that preserve the $L^2$ nature of the curvature of the connection are rather small. In the present case, however, all reducible connections are contained in a single orbit. It appears that the distinction between instanton connections and reducible connections for the {\sdyme} are, in this sense, similar to the distinction between harmonic maps of finite uniton number~\cite{UhlenbeckJDG} and harmonic maps of finite type. Since the original motivation for the current work (and~\cite{sdym}) was to investigate connections between integrable systems theory and Donaldson's use of the {\sdyme} in connection with four-dimensional topology~\cite{Don1}, it is rather striking that the behaviour of reducible connections and irreducible connections should be so different from the integrable systems point of view. Whether these results point to a deeper relationship between integrable systems theory and topological field theory would certainly seem worthy of further investigation.

\appendix
\section{Constant group action}
\label{orbits}

The group action $G(x, z) \mapsto h(x, z) G(x, z) h^*(x, z)$ is a little unusual. In order to gain some insight into this action, we consider some similar actions on simpler groups, analogous to the case where $G$ and $h$ are constant.

\subsection{$\SLtwoR$}
Consider the action of $\SLtwoR$ on itself given by
\[
\SLtwoR \times \SLtwoR \rightarrow \SLtwoR; \qquad \left( h, g \right) \mapsto h \cdot g := h g h^t,
\]
where ${}^t$ denotes transpose. The subgroup $\PSLtwoR \cong \mathrm{SO}^0_{2, 1}$ acts effectively. We decompose $g$ into symmetric and skew-symmetric parts
\[
g = U + \alpha \epsilon,
\]
where $U$ is symmetric, $\alpha \in \mathbb{R}$ and $\epsilon = \begin{pmatrix} 0 &1 \\ -1 &0 \end{pmatrix}$. It then follows that $\alpha$ is invariant under the action of $h$ and $U$ transforms according to $U \mapsto h \cdot U := h U h^t$. The fact that $g$ lies in {\SLtwoC} implies that
\[
\det U = 1 - \alpha^2.
\]
Writing
\[
U = \begin{pmatrix} t+x &y \\ y &t-x \end{pmatrix}
\]
then $\det U = - \| \mathbf{u} \|^2$, where $\mathbf{u} = (t, x, y) \in \mathbb{R}^{2, 1}$. The orbits of the group action are then parametrised by $\alpha \in \mathbb{R}$ and consist of vectors $\mathbf{u} \in \mathbb{R}^{2, 1}$ with
\[
\| \mathbf{u} \|^2 = \alpha^2 - 1.
\]
Since the restriction on $\mathbf{u}$ is insensitive to the sign of $\alpha$, we consider the orbits for $\alpha \ge 0$:
\vs
\noindent{\fbox{$\alpha = 0$}} Here there are two orbits consisting of symmetric elements of $\SLtwoR$. We have $\| \mathbf{u} \|^2 = -1$, so $\mathbf{u}$ lies on the two-sheeted hyperboloid in $\mathbb{R}^{2, 1}$, with each sheet constituting an orbit. In this case, giving the orbits the induced hyperbolic metric, the group {\SLtwoR} acts isometrically.

\vs
\noindent{\fbox{$0 < \alpha < 1$}} In this case, there are two orbits i.e. the two components of the hyperboloid  $\| \mathbf{u} \|^2 = -1 + \alpha^2 \in (-1, 0)$ in $\mathbb{R}^{2, 1}$. Again, the group {\SLtwoR} acts isometrically with respect to the induced metric on the orbits.

\vs
\noindent{\fbox{$\alpha = 1$}} In this case, $\| \mathbf{u} \|^2 = 0$, so either $\mathbf{u} = 0$ or $\mathbf{u}$ is null. In the first case, the group orbit consists of the point $\mathbf{u} = 0$. In the latter case, the future and past null-cones of the origin give two distinct group orbits.

\vs
\noindent{\fbox{$\alpha > 1$}} In this case, there is one orbit, consisting of the one-sheeted hyperboloid $\| \mathbf{u} \|^2 = -1 + \alpha^2 \in (1, \infty)$ in $\mathbb{R}^{2, 1}$. In this case, {\SLtwoR} acts isometrically with respect to the induced (Lorentzian) metric on the orbit.

\subsection{$\SLtwoC$}
In particular, we consider the action of $\SLtwoC$ on itself given by
\begin{equation}
\SLtwoC \times \SLtwoC \rightarrow \SLtwoC; \qquad \left( h, g \right) \mapsto h \cdot g := h g h^*,
\label{hghstar}
\end{equation}
where ${}^*$ denotes complex-conjugate transpose. The subgroup $\PSLtwoC \cong \mathrm{SO}_{3, 1}$ acts effectively. It is straightforward to check that
\[
I[g] := \frac{1}{2} \mathrm{tr} \left( g \left( g^{-1} \right)^{\dagger} \right)
\]
is invariant under the transformation $g \mapsto h \cdot g$. It is useful to split $g$ into Hermitian and skew-Hermitian parts
\[
g = U + V \quad \mbox{where} \quad U^* = U, \quad V^* = - V,
\]
and to note that this decomposition is preserved under~\eqref{hghstar} (i.e. $\left( h \cdot g \right)_U = h \cdot g_U$, etc.) A straightforward calculation implies that
\[
\det U = \frac{1}{2} \left( I + 1 \right).
\]
In particular, letting $U = t + x \tau_1 + y \tau_2 + z \tau_3$ and $\mathbf{u} := (t, x, y, z) \in \mathbb{R}^{3, 1}$, we deduce that
\begin{equation}
\| \mathbf{u} \|^2 = - \frac{1}{2} \left( I + 1 \right).
\label{usquared}
\end{equation}
Similarly, letting $V = i \left( T + X \tau_1 + Y \tau_2 + Z \tau_3 \right)$ and $\mathbf{v} := (T, X, Y, Z) \in \mathbb{R}^{3, 1}$, we find that
\[
\det g = - \| \mathbf{u} \|^2 - 2 i \langle \mathbf{u}, \mathbf{v} \rangle + \| \mathbf{v} \|^2.
\]
Since $g \in \SLtwoC$, we therefore deduce that
\begin{align}
\| \mathbf{v} \|^2 = - \frac{1}{2} \left( I - 1 \right),
\label{vsquared}
\\
\langle \mathbf{u}, \mathbf{v} \rangle = 0.
\label{udotv}
\end{align}

If $I > 1$ then $\mathbf{u}$ and $\mathbf{v}$ would be non-zero, orthogonal, time-like vectors in $\mathbb{R}^{3, 1}$. Since this cannot occur, we deduce that $I \le 1$. We investigate the distinct cases separately:

\vs
\noindent{\fbox{$I=1$}} In this case, $\| \mathbf{u} \|^2 = - 1$ and $\| \mathbf{v} \|^2 = 0$. As such, $\mathbf{u}$ lies on the two-sheeted hyperboloid in $\mathbb{R}^{3, 1}$. The condition that $\| \mathbf{v} \|^2 = 0$ and is orthogonal to the non-zero, time-like vector $\mathbf{u}$ then implies that $\mathbf{v} = 0$. As such, we have two distinct orbits, corresponding to the two components of the two-sheeted hyperboloid. These orbits correspond to the Hermitian elements of {\SLtwoC}. Giving the orbits the hyperbolic metric induced from $\mathbb{R}^{3, 1}$, the group {\SLtwoC} acts isometrically.

\vs
Note that for $I < 1$, the vector $\mathbf{v}$ is always space-like, and lies on the one-sheeted hyperboloid $\Sigma_I := \left\{ \mathbf{w} \in \mathbb{R}^{3, 1}: \| \mathbf{w} \|^2 = \frac{1}{2}(1-I) \right\}$ in $\mathbb{R}^{3, 1}$.

\vs
\noindent{\fbox{$-1 < I < 1$}} We have $\| \mathbf{u} \|^2 = - (I+1)/2$ in $\mathbb{R}^{3, 1}$. Since $\mathbf{u}$ is orthogonal to $\mathbf{v}$, we may view $\mathbf{u}$ as a time-like vector of length $\sqrt{(I+1)/2}$ lying in the two-sheeted hyperboloid in $T_{\mathbf{v}} \Sigma_I$. As such, we have two distinct orbits, consisting of the two components of the two-sheeted hyperboloid bundle in $T \Sigma_I$. In this case, the group action on the orbit is the action induced by the isometric action of {\SLtwoR} on the Lorentzian metric induced on the one-sheeted hyperboloid.

Alternatively, we may view $\mathbf{u}$ as a time-like vector lying on the two-sheeted hyperboloid $\| \mathbf{u} \|^2 = - (I+1)/2$ in $\mathbb{R}^{3, 1}$. We then view $\mathbf{v}$ as a tangent vector to the hyperboloid of length $\frac{1}{\sqrt{2}} \left( 1 + |I| \right)$. Therefore the orbits in this case may be identified with the radius $\frac{1}{\sqrt{2}} \left( 1 + |I| \right)$ sphere sub-bundle of the tangent bundle of the hyperbolic space of radius $\frac{1}{\sqrt{2}} \left( I + 1 \right)$. Again, there are two orbits corresponding to the two components of the hyperboloid. In this case, the group action on the orbit is the action induced by the isometric action of {\SLtwoR} on the induced metric on the two-sheeted hyperboloid.

\vs
\noindent{\fbox{$I = -1$}} In this case, $\| \mathbf{u} \|^2 = 0$ and $\| \mathbf{v} \|^2 = 1$. As such, we may view $\mathbf{u}$ as a null vector in $T_{\mathbf{v}} \Sigma_{-1}$. There are then three distinct orbits. The first consists of $\mathbf{u} = 0$, and is simply the hyperboloid $\Sigma_{-1}$. This orbit consists of the skew-Hermitian elements of {\SLtwoC}. The other orbits consist of the sub-bundle of $T \Sigma_{-1}$ consisting of the past and future null cone of the origin in each tangent space.

\vs
\noindent{\fbox{$I < -1$}} In this case, $\| \mathbf{u} \|^2 = (|I|+1)/2 > 0$ in $\mathbb{R}^{3, 1}$. Therefore there is one orbit, consisting of the one-sheeted hyperboloid sub-bundle of $T \Sigma_1$. The {\SLtwoC} action is that induced by the isometric action on the induced Lorentzian metric on $\Sigma_I$.

\def\cprime{$'$}

\end{document}